\documentclass[conference, 10pt]{ieeeconf}

\IEEEoverridecommandlockouts     

\usepackage{cite}
\usepackage{amssymb,amsfonts,amsmath,amsthm,amscd}
\usepackage[all]{xy}
\usepackage{graphicx}

\newtheorem{theorem}{Theorem}
\newtheorem{lemma}{Lemma}
\newtheorem{corollary}{Corollary}
\newtheorem{proposition}{Proposition}

\newcommand{\E}{\ensuremath{\mathbb E}}
\newcommand{\ind}[1]{{\bf 1}_{\{#1\}}}
\newcommand{\deq}{\triangleq}
\def\wh#1{\ensuremath{\hat{#1}}}

\title{\Large \bf Directed Information and Pearl's Causal Calculus}
\author{Maxim Raginsky
\thanks{This work was supported by NSF grant CCF-1017564 and by AFOSR grant FA9550-10-1-0390.}
\thanks{The author is with the Department of Electrical and Computer Engineering, Duke University, Durham, NC. E-mail: {m.raginsky@duke.edu}.}}

\begin{document}

\maketitle

\begin{abstract}
Probabilistic graphical models are a fundamental tool in statistics, machine learning, signal processing, and control. When such a model is defined on a directed acyclic graph (DAG), one can assign a partial ordering to the events occurring in the corresponding stochastic system. Based on the work of Judea Pearl and others, these DAG-based ``causal factorizations" of joint probability measures have been used for characterization and inference of functional dependencies (causal links). This mostly expository paper focuses on several connections between Pearl's formalism (and in particular his notion of ``intervention") and information-theoretic notions of causality and feedback (such as causal conditioning, directed stochastic kernels, and directed information). As an application, we show how conditional directed information can be used to develop an information-theoretic version of Pearl's ``back-door'' criterion for identifiability of causal effects from passive observations. This suggests that the back-door criterion can be thought of as a causal analog of statistical sufficiency.
\end{abstract}

\thispagestyle{empty}
\pagestyle{empty}

\section{Introduction}
\label{sec:intro}

The problems of causality in engineered and natural systems have recently attracted the attention of information theorists and signal processing researchers \cite{RaoMotifDiscovery,MathaiGeneNetworks,RaoBioNetworks,AmblardNeuro,QuinnNeuro,QuinnCausalTrees}. The well-worn but nonetheless true maxim stating that ``correlation does not imply causation'' means that causal relationships cannot be captured by standard information-theoretic quantities like mutual information, conditional entropy, or divergence, because all of these are measures of {\em statistical dependence} (i.e., correlation). The first information-theoretic studies of causality were concerned with feedback communication systems and led to the development of the notion of {\em directed information} by Massey \cite{MasseyDirInfo}, with subsequent extensions and generalizations by Kramer, Tatikonda, and Mitter \cite{KramerThesis,TatikondaThesis,TatikondaMitter}. Connections between directed information and sequential prediction, source coding, and hypothesis testing have also been extensively investigated \cite{VenkataramananFeedforward,PermuterTrapdoor,GorantlaColemanCausal,PermuterDirInfo}.

However, causality has also been the subject of vigorous study in the statistics, artificial intelligence, and machine learning communities \cite{PearlPRIS,SpirtesBook,PearlCausalitySurvey,PearlCausality}. The key idea advanced in these works, particularly by Pearl, is that causality is synonymous with {\em functional} (rather than statistical) dependence. In other words, causal relationships correspond to stable deterministic mechanisms, by which one set of variables (the causes), together with some possibly unobserved exogenous disturbances, may affect another set of variables (the effects). Thus, inferring causal relationships requires {\em active experimentation} that {\em intervenes} into some of these mechanisms. In very schematic terms (this discussion will be made precise in the sequel), an ideal setting for identifying or estimating the ``causal effect'' of one observable (say, $X$) on another (say, $Y$) would permit the experimenter to disconnect $X$ from all mechanisms that influence it, force $X$ to take on some value(s) of interest, and then to estimate the probability distribution of $Y$ as a result of this intervention, while controlling for all possible spurious influences and factors. This is quite different from estimating the {\em statistical} effect of $X$ on $Y$, i.e., the conditional distribution $P_{Y|X}$, by means of passive observations, e.g., from a large number of independent samples from the joint distribution of $X,Y$.

The purpose of this mostly expository paper is to introduce the information theory, control, and signal processing communities to several key concepts of the probabilistic theory of causality and, along the way, to elucidate several connections between Pearl's treatment of interventions on the one hand, and information-theoretic concepts pertaining to causality (such as directed information \cite{MasseyDirInfo}, causal conditioning \cite{KramerThesis}, or directed stochastic kernels \cite{TatikondaThesis,TatikondaMitter}) on the other. In particular, the representation of causal relationships by Markov factorizations of joint probability distributions w.r.t.\ directed acyclic graphs (DAGs) \cite{PearlPRIS,SpirtesBook,PearlCausalitySurvey,PearlCausality}, such that the natural partial ordering of the vertices of the DAG corresponds to causal ordering of the events in the system under consideration, should be very congenial to systems theorists, who naturally think in terms of block diagrams, interconnections, and sequential recursive models.

Let us give a brief overview of the remainder of the paper. We first motivate the functional view of causality in Section~\ref{sec:causality} by means of a simple example of a point-to-point communication system. Next, in Section~\ref{sec:SDM_causality}, we develop the general framework for studying causality in {\em Markovian dynamical systems}. In particular, we motivate Pearl's definition of intervention as ``surgery'' on a sequential recursive representation of such a system, whereby the relations defining the intervened-upon variables are deleted, and all instances of these variables in the remaining relations are assigned to some fixed value. This operation has a natural diagrammatic representation on the DAG inducing the Markov factorization of the joint probability distribution of the system observables according to the sequential model. We also show that the probability distributions induced by this operation (i.e., what Pearl calls the {\em causal effects}) are in one-to-one correspondence with the directed stochastic kernels of Tatikonda and Mitter \cite{TatikondaThesis,TatikondaMitter}. This correspondence is then used in Section~\ref{sec:dir_info} to show how directed information (and certain generalizations, such as conditional directed information) can be used to quantify the strength of causal effects by comparing them with ordinary (observational) conditional distributions. Section~\ref{sec:back_door} develops an information-theoretic interpretation of Pearl's ``back-door'' criterion \cite[Sec.~3.3.1]{PearlCausality} (a sufficient condition for identifiability of causal effects from observational data) in terms of conditional directed information, showing in effect that the back-door criterion can be viewed as a natural causal analog of statistical sufficiency.

\section{Revealing causality through functional dependence}
\label{sec:causality}

To illustrate the difference between statistical dependence and causal dependence, consider the standard diagram of a point-to-point communication system without feedback, as shown in Figure~\ref{fig:comm_sys}. A message $W$ is mapped into a channel input symbol $X = e(W)$, $X$ is transmitted over a channel with transition kernel $P_{Y|X}$, and the resulting channel output symbol $Y$ is processed at the receiver into a decoded message $\tilde{W} = d(Y)$, where $e$ and $d$ are some deterministic encoding and decoding functions.

\begin{figure}
	\centerline{\includegraphics[width=0.8\columnwidth]{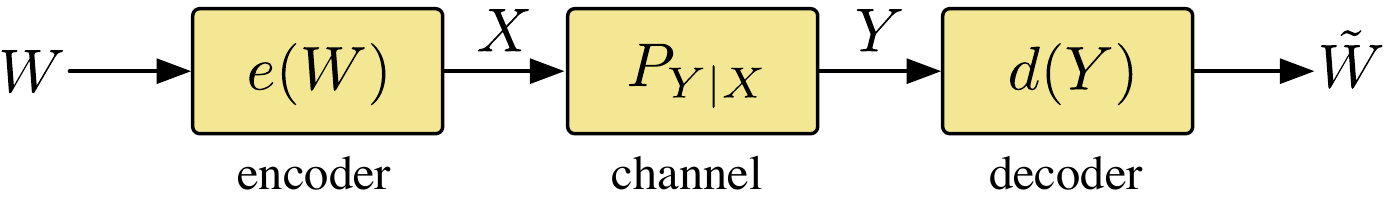}}
	\caption{\label{fig:comm_sys}A generic communication system without feedback.}
\end{figure}

It is intuitively clear that the message $W$ ``causes'' the decoded message $\tilde{W}$ and not the other way around, but we cannot tell this from the joint distribution of $W$, $X$, $Y$, and $\tilde{W}$. Indeed, we have
\begin{align*}
	&P_{WXY\tilde{W}}(w,x,y,\tilde{w}) \\
	& \quad = P_W(w) \ind{e(w)=x} P_{Y|X}(y|x)\ind{d(y) = \tilde{w}},
\end{align*}
so that the joint distribution of $W$ and $\tilde{W}$, given by
\begin{align*}
	P_{W\tilde{W}}(w,\tilde{w}) &= P_W(w)\sum_{x,y}\ind{e(w)=x} P_{Y|X}(y|x)\ind{d(y) = \tilde{w}} \\
	&= P_W(w)\sum_y P_{Y|X}(y|e(w)) \ind{d(y) = \tilde{w}} \\
	&\equiv P_W(w) P_{\tilde{W}|W}(\tilde{w}|w),
\end{align*}
can also be factored as $P_{W\tilde{W}}(w,\tilde{w}) = P_{\tilde{W}}(\tilde{w})P_{W|\tilde{W}}(w|\tilde{w})$, which merely shows that $W$ and $\tilde{W}$ are statistically dependent on one another. Indeed, to quote Massey \cite{MasseyDirInfo}, ``statistical dependence, unlike causality, has no inherent directivity.'' If the encoder, the channel, and the decoder are nondegenerate, so that $I(W; \tilde{W}) > 0$, then the dependence between the message $W$ and the decoded message $\tilde{W}$ is completely symmetric: $W$ depends on $\tilde{W}$, and $\tilde{W}$ depends on $W$.

In order to elicit the {\em causal} influence of the transmitted message on the decoded message, as well as the lack of causal influence in the opposite direction, we need to break this symmetry. To that end, let us represent the stochastic transformation $X \to Y$ effected by the channel $P_{Y|X}$ as a {\em deterministic} mapping $Y = f(X,U)$, where $U$ is random {\em channel noise}, assumed to be independent of $W$ and $X$. (Indeed, any stochastic kernel $P_{Y|X}$ can be represented in this form for a suitable choice of $f$ and $P_U$.) This representation is shown in Figure~\ref{fig:comm_sys_latent}.

\begin{figure}
	\centerline{\includegraphics[width=0.8\columnwidth]{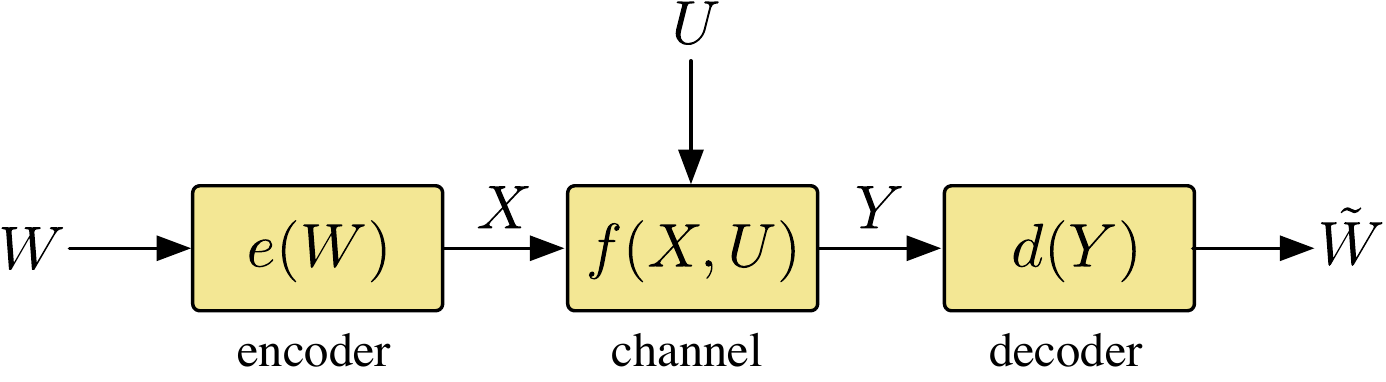}}
	\caption{\label{fig:comm_sys_latent}An equivalent diagram of the system in Figure~\ref{fig:comm_sys}.}
\end{figure}

Now we can represent our communication system in the following {\em sequential form}:
\begin{align}
	W & \sim P_W \nonumber\\
	U & \sim P_U \nonumber \\
	X & = e(W) \label{eq:channel_SDM}\\
	Y & = f(X,U)\nonumber \\
	\tilde{W} &= d(Y) \nonumber
\end{align}
What happens if we make a hard assignment $W \leftarrow w$ of a specific value $w$ to the transmitted message? Looking at the sequential model in \eqref{eq:channel_SDM}, we see that this action will influence the ``downstream'' variables $U,X,Y,\tilde{W}$ as follows:
\begin{align*}
	U & \sim P_U \\
	X & = e(w) \\
	Y & = f(e(w),U) \\
	\tilde{W} &= d(f(e(w),U)).
\end{align*}
The corresponding joint distribution of $U$, $X$, $Y$ and $\tilde{W}$ resulting from the action $W \leftarrow w$, which we will denote by $P_{UXY\tilde{W}|W\leftarrow w}$, has the form
\begin{align*}
	&P_{UXY\tilde{W}|W\leftarrow w}(u,x,y,\tilde{w})\\
	 &= P_U(u)\ind{e(w)=x}\ind{f(e(w),u)=y}\ind{d(f(e(w),u))=\tilde{w}}.
\end{align*}
Marginalizing out the channel noise $U$, the channel input $X$, and the channel output $Y$, we get
\begin{align*}
	P_{\tilde{W}|W\leftarrow w}(\tilde{w}) &= \sum_{u,y}P_U(u)\ind{f(e(w),u)=y}\ind{d(f(e(w),u))=\tilde{w}}.
\end{align*}
This distribution is, in fact, equal to the ordinary conditional distribution $P_{\tilde{W}|W=w}$, given by
\begin{align*}
	P_{\tilde{W}|W=w}(\tilde{w})
	&= \sum_y P_{Y|X}(y|e(w))\ind{d(y)=\tilde{w}} \\
	&= \sum_{u,y}P_U(u)\ind{f(e(w),u)=y}\ind{d(f(e(w),u))=\tilde{w}}.
\end{align*}
Again, assuming that the mappings $e$, $f$, $d$ are nondegenerate, there exist at least two values $w,w'$ for the transmitted message, for which $P_{\tilde{W}|W=w} \neq P_{\tilde{W}|W=w'}$ and, consequently, $P_{\tilde{W}|W \leftarrow w} \neq P_{\tilde{W}|W \leftarrow w'}$. In other words, the downstream effect of the hard assignment $W \leftarrow w$ is different from that of $W \leftarrow w'$.

Now let us consider what happens if we make a hard assignment $\tilde{W} \leftarrow \tilde{w}$ of the {\em decoded} message. One way to do this would be to replace the original decoding map $d$ with the constant map $d_{\tilde{w}}(y) = \tilde{w}$ for all $y$. The effect of this hard assignment on the remaining variables can be represented as
\begin{align*}
	W & \sim P_W \\
	U & \sim P_U \\
	X & = e(W) \\
	Y & = f(e(W),U)
\end{align*}
This clearly shows that the joint distribution of the ``upstream'' random variables $W,U,X,Y$ is unaffected by the action $\tilde{W} \leftarrow \tilde{w}$; in fact, exactly the same conclusion would hold if we replaced the original decoding map $d$ with any other decoding map $d'$. In other words,
\begin{align*}
	P_{WUXY|\tilde{W}\leftarrow \tilde{w}}= P_{WUXY}, \quad P_{W|\tilde{W} \leftarrow \tilde{w}} = P_W,
\end{align*}
which shows the absence of causal influence of $\tilde{W}$ on $W$.

\section{Causality in sequential dynamical systems}
\label{sec:SDM_causality}

The simple example of the preceding section illustrates the general treatment of causality advocated by Pearl. To motivate it, let us consider a stochastic dynamical system with multiple feedback loops and exogenous influences (or disturbances) shown in Figure~\ref{fig:SDM}. The exogenous disturbances are modeled by $n$ random variables $U_1,\ldots,U_n$ with a fixed joint distribution $P_{U^n} = P_{U_1\ldots U_n}$, while the system observables are represented by $n$ variables $X_1,\ldots,X_n$, related to $U^n$ and to one another by $n$ coupled equations
\begin{align}\label{eq:GDM}
	X_i = f_i(X^n,U^n), \qquad i \in [n]
\end{align}
We assume that the system specification is sound in the sense that the equations \eqref{eq:GDM} have a unique solution $X^n=x^n$ for any realization $U^n = u^n$ of the exogenous variables. This representation of stochastic dynamical systems as multiple feedback loops was used by Witsenhausen \cite{WitsenhausenInfoStruct,WitsenhausenSep,WitsenhausenStandardSequential} in his seminal work on distributed control systems.

\begin{figure}
\centerline{\includegraphics[width=0.5\columnwidth]{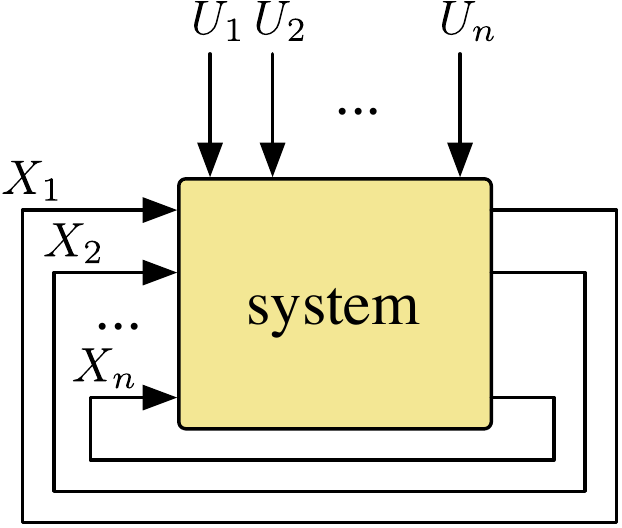}}
	\caption{\label{fig:SDM} A generic stochastic dynamical system with multiple feedback loops and exogenous disturbances.}
\end{figure}

This description allows for arbitrary dependencies between the observables $X_1,\ldots,X_n$, including cycles of the form $X_j = f_j(X_i,U_j)$, $X_k = f_k(X_j,U_k)$, $X_i = f_i(X_k,U_i)$. In order to study causality, we will limit ourselves to {\em sequential} dynamical systems, in which the observables $X_1,\ldots,X_n$ are ordered in such a way that, for each $i \in [n]$, there exists a set $\Pi_i \subseteq [i-1]$, such that the function $f_i$ depends essentially only on $X^{\Pi_i} \deq (X_j : j \in \Pi_i)$ and on $U_i$:
\begin{align}\label{eq:SDM}
	X_i = f_i(X^{\Pi_i},U_i), \qquad i \in [n]
\end{align}
Moreover, if for each $i$ the exogenous variable $U_i$ is independent of $(X^{i-1},U^{i-1})$, then the sequential model \eqref{eq:SDM} specifies the joint distribution $P_{X^n}$ via the {\em Markov factorization}
\begin{align}\label{eq:SDM_Markov}
	P_{X^n}(x^n) &= \prod^n_{i=1}P_{X_i|X^{\Pi_i}}(x_i|x^{\Pi_i}),
\end{align}
where, for each $i \in [n]$, 
\begin{align}
	P_{X_i|X^{\Pi_i}}(x_i|x^{\Pi_i}) = P_{U_i}\left( f_i(x^{\Pi_i},U_i) = x_i \right),
\end{align}
and $X^{[i-1]\backslash \Pi_i} \to X^{\Pi_i} \to X_i$ is a Markov chain. We will refer to any stochastic dynamical system specified by \eqref{eq:SDM} with independent disturbances $U_1,\ldots,U_n$ as a {\em Markovian dynamical system}. Apparently, one of the earliest attempts to study causality by means of simple Markovian models of this sort was made in the 1920's by the geneticist Sewall Wright \cite{WrightGenetics}.

The Markov factorization \eqref{eq:SDM_Markov} can also be represented in graphical form by means of a directed graph with $n$ vertices, where vertex $i$ is associated with $X_i$, and there is a directed edge from vertex $j$ to vertex $i$ if and only if $j \in \Pi_i$. Because $\Pi_i \subseteq [i-1]$, we end up with a DAG. Since we will use this graphical representation rather heavily in the sequel, let us pause to define some concepts associated with DAGs. Given $i \in [n]$, we let $\Delta_i \subset [n]$ denote the set of all {\em descendants} of $i$, i.e., the set of all $j \in [n]\backslash \{i\}$, such that there is a directed path from $i$ to $j$.  Similarly, we let $A_i$ denote the set of all {\em ancestors} of $i$, i.e., all $j \in [n]\backslash \{i\}$ connected to $i$ by directed paths. We also let  $\Delta^+_i \deq \Delta_i \cup \{i\}$, so that $N_i \deq [n] \backslash \Delta^+_i$ is the set of all {\em nondescendants} of $i$.   Note that
\begin{align}\label{eq:ancestral_inclusion}
 \bigcup_{j \in N_i}A_j \subset N_i.
\end{align}
Indeed, if for some $j \in N_i$ there exists some $k \in A_j \cap \Delta^+_i$, then there is a directed path from $i$ to $j$ going through $k$, which is impossible by the definition of $N_i$.

\subsection{Interventions in Markovian dynamical systems}

Consider a Markovian dynamical system specified according to \eqref{eq:SDM}. Just as we did in the simple example of Section~\ref{sec:causality}, we can study the causal effect of one set of variables $X^S$, $S \subset [n]$, on another set $X^T$ with $S \cap T = \varnothing$ by examining the impact of hard assignments of the form $X^S \leftarrow x^S$ on $X^T$. The main idea is to start with the recursive representation \eqref{eq:SDM}, delete all equations defining the variables $X_i, i \in S$, and replace all other instances of these variables with the assigned values. For example, the effect of what Pearl calls an {\em atomic intervention} $X_i \leftarrow x_i$ can be represented as the following modification of \eqref{eq:SDM}:
\begin{align}\label{eq:SDM_interv_i}
	X_j &= \begin{cases}
	f_j(X^{\Pi_j},U_j)\big|_{X_i = x_i}, & \text{if } j \in \Delta_i \\
	f_j(X^{\Pi_j},U_j), & \text{if } j \in N_i
\end{cases}
\end{align}
Now, for any set $T \subseteq [n]\backslash \{i\}$, let $P_{X^T|X_i \leftarrow x_i}$ denote the probability distribution of $X^T$ induced by the modified model \eqref{eq:SDM_interv_i}. Other notation used by Pearl and coauthors includes $P_{X^T|\wh{X}_i = \wh{x}_i}$ (where hats are added to the intervened-upon variables and the values assigned to them) and $P_{X^T|\text{do}(X_i = x_i)}$; we will use some of these interchangeably. The main claim is that these {\em interventional distributions} describe the {\em causal} effect of $X_i$ upon $X^T$. Let us see some illustrations in support of this claim.

First of all, we would intuitively expect that the intervention $X_i \leftarrow x_i$ would only affect the descendants of $i$. This is indeed true:

\begin{lemma}\label{lm:nondescendants} For any $T \subseteq N_i$ and any intervention $X_i \leftarrow x_i$,
	\begin{align*}
		P_{X^T|X_i \leftarrow x_i} = P_{X^T},
	\end{align*}
	where the distribution $P_{X^T}$ on the right-hand side is induced by the original model \eqref{eq:SDM}.
\end{lemma}
\begin{proof} Because of \eqref{eq:ancestral_inclusion}, no $X_k$ with $k \in \Delta^+_i$ appears in any of the equations defining $X^{N_i}$ in \eqref{eq:SDM_interv_i}. Hence, the joint distribution of $X^{N_i}$ in the modified model \eqref{eq:SDM_interv_i} is the same as in the original model \eqref{eq:SDM}.
\end{proof}
\noindent Since $\Pi_i \subseteq N_i$, we have
\begin{corollary} For any $i \in [n]$ and any intervention $X_i \leftarrow x_i$,
	\begin{align*}
		P_{X^{\Pi_i}|X_i \leftarrow x_i} = P_{X^{\Pi_i}}.
	\end{align*}
\end{corollary}
\noindent The extension to multiple interventions of the form $X^S \leftarrow x^S$ is immediate: defining the sets
\begin{align*}
	\Delta_S \deq \bigcup_{i \in S}\Delta_i,\quad \Delta^+_S \deq \Delta_S \cup S,\quad N_S \deq [n]\backslash \Delta^+_S
\end{align*}
we can represent the effect of the intervention $X^S \leftarrow x^S$ on $X^{S^c} = (X_j : j \not\in S)$ by
\begin{align}\label{eq:SDM_interv_S}
	X_j &= \begin{cases}
	f_j(X^{\Pi_j},U_j)\big|_{X^S = x^S}, & \text{if } j \in \Delta_S \\
	f_j(X^{\Pi_j},U_j), & \text{if } j \in N_S
\end{cases}
\end{align}
and, for any $T \subset [n]\backslash S$, the interventional distribution $P_{X^T|X^S \leftarrow x^S}$ is given by the joint distribution of $X^T$ induced by \eqref{eq:SDM_interv_S}. Going through the same reasoning as before, we obtain the following generalization of Lemma~\ref{lm:nondescendants}:
\begin{lemma}\label{lm:nondescendants_2} For any $S \subseteq [n]$, any $T \subseteq N_S$, and any intervention $X^S \to x^S$,
	\begin{align*}
		P_{X^T|X^S \leftarrow x^S} = P_{X^T}.
	\end{align*}
\end{lemma}
\noindent On the other hand, let us pick some $i \in [n]$ and consider the causal effect of the intervention $X^{\Pi_i} \leftarrow x^{\Pi_i}$ upon $X_i$:
\begin{lemma}\label{lm:parents}
For any $S \subset [n]$ and any intervention $X^{\Pi_S} \leftarrow x^{\Pi_S}$, we have
\begin{align*}
	P_{X_S|X^{\Pi_S} \leftarrow x^{\Pi_S}} = P_{X^S|X^{\Pi_S}=x^{\Pi_S}}.
\end{align*}	
Moreover, for any $T \subseteq (S \cup \Pi_S)^c$ and any intervention $X^T \leftarrow x^T$, we have
\begin{align*}
	P_{X^S|X^{\Pi_S} \leftarrow x^{\Pi_S}, X^T \leftarrow x^T} = P_{X^S|X^{\Pi_S} \leftarrow x^{\Pi_S}} = P_{X^S|X^{\Pi_S}=x^{\Pi_S}}.
\end{align*}
\end{lemma}
\begin{proof} Observe that, as a result of the intervention $X^{\Pi_S} \leftarrow x^{\Pi_S}$, we have
	\begin{align*}
		X_j = f_j(x^{\Pi_j},U_j), \qquad \forall j \in S
	\end{align*}
which means that, for any $x^S$ and any additional intervention $X^T \leftarrow x^T$, where $T$ is disjoint form $S \cup \Pi_S$, we have
\begin{align*}
&	P_{X^S|X^{\Pi_S}\leftarrow x^{\Pi_S},X^T \leftarrow x^T}(x^S) \\
& \qquad = P_{U^S}\left( f_j(x^{\Pi_j},U_j) = x_j, \forall j \in S\right) \\
& \qquad = P_{X^S|X^{\Pi_S}\leftarrow x^{\Pi_S}}(x^S) \\
& \qquad = P_{X^S|X^{\Pi_S}=x^{\Pi_S}}(x^S).
\end{align*}
In other words, the joint distribution of $X^S$ induced by \eqref{eq:SDM_interv_S} is unaffected by $X^T \leftarrow x^T$.
\end{proof}
In terms of the Markov factorization \eqref{eq:SDM_Markov}, we can express the interventional distributions $P_{X^T|X^S \leftarrow x^S}$ for any pair of disjoint sets $S,T \subset [n]$ as follows. First, we write down the ``global" interventional distribution of $X^{S^c}$ given the action $X^S \leftarrow x^S$,
\begin{align}\label{eq:SDM_Markov_interv}
	P_{X^{S^c}|X^S \leftarrow x^S}(x^{S^c}) &= \prod_{i \in S^c} P_{X_i|X^{\Pi_i}}(x_i|x^{\Pi_i}),
\end{align}
and then marginalize out all variables outside of $T$:
\begin{align}\label{eq:SDM_Markov_interv_marginal}
	P_{X^T|X^S\leftarrow x^S}(x^T) &= \sum_{x^{S^c \cap T^c}} P_{X^{S^c}|X^S \leftarrow x^S}(x^{S^c})
\end{align}
Note that, in general, this is different from the ordinary conditional distribution $P_{X^T|X^S  = x^S}$, which has the following standard interpretation in Bayesian terms: Suppose we can only observe $X^S$, but not $X^{S^c}$. If we let system evolve freely according to \eqref{eq:SDM} and then {\em observe} that $X^S = x^S$, then $P_{X^T|X^S = x^S}$ represents our {\em posterior beliefs} about $X^T$ based on the {\em observed evidence} $X^S = x^S$.

\subsection{Interventions in graphical models}
\label{ssec:graphical_models}

Graphical model representations of Markovian dynamical systems offer a convenient visual way of computing interventional distributions. Essentially, if we wish to write down the interventional distribution $P_{X^{S^c}|X^S \leftarrow x^S}$, we draw the corresponding DAG, remove all edges incident upon the vertices in $S$, and write down the joint distribution of $X^{S^c}$ induced by the resulting DAG, while setting $X^S$ to the assigned values $x^S$. 

Let us see this on a couple of examples. Consider the following graphical model:
\begin{align*}
	\xymatrix{
	X_1 \ar[rr]\ar[dr]&  & X_4 \ar[dr] & \\
	    & X_3\ar[rr]\ar[dr]    &     & X_6 \\
	X_2\ar[ur]&       & X_5\ar[ur]
	}
\end{align*}
It specifies the joint distribution of $X^6 = (X_1,\ldots,X_6)$ via
\begin{align*}
	P_{X^6}(x^6) &= P_{X_1}(x_1)P_{X_2}(x_2)\\
	& \qquad \times P_{X_3|X^2}(x_3|x^2)P_{X_4|X_1}(x_4|x_1) \\
	& \qquad \times P_{X_5|X_3}(x_5|x_3) P_{X_6|X^5_3}(x_6|x^5_3).
\end{align*}
The effect of the intervention $X_3 \leftarrow x_3$ can be represented graphically as follows:
\begin{align*}
	\xymatrix{
	X_1 \ar[rr]&  & X_4 \ar[dr] & \\
	    & *+[F]{X_3}\ar[rr]\ar[dr]    &     & X_6 \\
	X_2&   *+[o]{x_3} \ar[u]^{=}      & X_5\ar[ur]
	}
\end{align*}
In other words, the intervened-upon variable $X_3$, which is enclosed in a box, is disconnected from its direct causes in $\Pi_3 = \{1,2\}$, and an additional arrow is added to indicate the hard assignment $X_3 \leftarrow x_3$. The resulting interventional distribution $P_{X_1,X^6_2|X_3 \leftarrow x_3}$ can be read off directly from the diagram:
\begin{align*}
	P_{X_1,X^6_2|X_3 \leftarrow x_3}(x_1,x^6_2) &= P_{X_1}(x_1)P_{X_2}(x_2) \\
	& \qquad \times P_{X_4|X_1}(x_4|x_1) P_{X_5|X_3}(x_5|x_3) \\
	& \qquad \times  P_{X_6|X^5_3}(x_6|x^5_3).
\end{align*}
As another example, consider the following diagram, which depicts communication over a discrete memoryless channel $P_{Y|X}$ using a sequence of possibly randomized feedback encoders $P_{X_{i}|X_{i-1},Y_{i-1}}, i \in [n]$:
\begin{align*}
\xymatrix{
X_1 \ar[r]\ar[d] & X_2 \ar[r]\ar[d] & X_3 \ar[r]\ar[d] & \,\,{}_{}\ldots \ar[r] & X_n \ar[d] \\
Y_1 \ar[ur] & Y_2 \ar[ur] & Y_3 \ar[ur] & {}_{}\ldots\,\,\,{}_{}\,\,\ar[ur] & Y_n
}
\end{align*}
The effect of the intervention $Y_1 \leftarrow y_1, \ldots, Y_n \leftarrow y_n$ is represented graphically as
\begin{align*}
\xymatrix{
X_1 \ar[r] & X_2 \ar[r] & X_3 \ar[r] & {}_{\,\,} \ldots \ar[r] & X_n  \\
*+[F]{Y_1} \ar[ur] & *+[F]{Y_2} \ar[ur] & *+[F]{Y_3} \ar[ur] & {}_{\,\,} \ldots\,\,\,{}_{} \ar[ur]  & *+[F]{Y_n} \\
y_1 \ar[u]^= & y_2 \ar[u]^= & y_3 \ar[u]^= & \ldots & y_n  \ar[u]^=
}
\end{align*}
and the corresponding interventional distribution is
\begin{align*}
	P_{X^n|Y^n\leftarrow y^n}(x^n) &= \prod^n_{i=1} P_{X_i|X_{i-1},Y_{i-1}}(x_i|x_{i-1},y_{i-1}).
\end{align*}

\subsection{Interventional distributions as directed stochastic kernels}

As it turns out, Pearl's construction of interventional distributions has been developed independently by Tatikonda and Mitter \cite{TatikondaThesis,TatikondaMitter} under the name of {\em directed stochastic kernels} in their work on the capacity of channels with feedback. 

Tatikonda and Mitter consider an $n$-tuple of causally ordered random variables $X_1,\ldots,X_n$ with joint distribution
\begin{align*}
	P_{X^n}(x^n) &= \prod^n_{i=1}P_{X_i|X^{i-1}}(x_i|x^{i-1})
\end{align*}
(of course, we are free to factor $P_{X^n}$ along any other ordering of the variables, but the subsequent definitions depend on a fixed ordering). Then for any $S \subset [n]$ they define the {\em directed stochastic kernel} $\vec{P}_{X^{S^c}|X^S = x^S}$ by
\begin{align}\label{eq:DSK}
	\vec{P}_{X^{S^c}|X^S=x^S}(x^{S^c}) \deq \prod_{i \in S^c}P_{X_i|X^{i-1}}(x_i|x^{i-1}).
\end{align}
It is easy to see that this definition is equivalent to Pearl's. Indeed, if we consider the DAG with $n$ vertices that has $\Pi_i = [i-1]$ for each $i \in [n]$, then $\vec{P}_{X^{S^c}|X^S = x^S}$ defined in \eqref{eq:DSK} is equal to $P_{X^{S^c}|X^S\leftarrow x^S}$ defined in \eqref{eq:SDM_Markov_interv}. Conversely, if the variables $X_1,\ldots,X_n$ are ordered in such a way that for each $i \in [n]$ there exists some $\Pi_i \subseteq [i-1]$ such that $X^{[i-1]\backslash \Pi_i} \to X^{\Pi_i} \to X_i$ is a Markov chain, then
\begin{align*}
	P_{X^{S^c}|X^S \leftarrow x^S}(x^{S^c}) &= \prod_{i \in S^c}P_{X_i|X^{\Pi_i}}(x_i|x^{\Pi_i}) \\
	&= \prod_{i \in S^c} P_{X_i|X^{i-1}}(x_i|x^{i-1}) \\
	&= \vec{P}_{X^{S^c}|X^S \leftarrow x^S}(x^{S^c}),
\end{align*}
where the first step uses \eqref{eq:SDM_Markov_interv}, and the second uses \eqref{eq:DSK} and the above Markov chain condition.

\subsection{Interventions as channels}

The interventional distribution $P_{X^T| X^S \leftarrow x^S}$ can be viewed as a mapping from the set of all tuples $x^S = (x_i : i \in S)$ into the set of all probability distributions for $X^T$. Any such mapping defines a {\em channel} \cite{Dob59} with input variable $X^S$ and output variable $X^T$. If $S = \Pi_T$, then Lemma~\ref{lm:parents} shows that this channel coincides with the specification of the conditional distribution of $X^T$ given $X^{\Pi_T}$ in the intervention-free system. This equality of the originally prescribed stochastic kernels and the directed stochastic kernels holds whenever $X^S$ (resp., $X^T$) is the complete input (resp., output) variable of an encoder, decoder, or controller. By contrast, whenever $P_{X^T|X^S \leftarrow x^S} \neq P_{X^T|X^S = x^S}$ for some $x^S$, we can conclude that there are some additional causal or statistical relationships between $X^S$ and $X^T$.

\section{Directed information as a measure of causality}
\label{sec:dir_info}

Now that we have motivated the notion of a causal effect, we can proceed to define various information-theoretic quantities that capture causality as opposed to dependence. Assuming, as before, a Markovian dynamical system of the form \eqref{eq:SDM}, let us consider the interventional distribution $P_{X^T|\wh{X}^S}(\cdot|\wh{x}^S)$ for disjoint sets $S,T \subset [n]$. As we have pointed out already, this distribution is, in general, different from the conditional distribution $P_{X^T|X^S}(\cdot|x^S)$. In particular, if $P_{X^T|\wh{X}^S}(\cdot|\wh{x}^S) = P_{X^T}(\cdot)$ for any intervention $X^S \leftarrow x^S$, then the variables in $S$ have no causal influence on those in $T$. On the opposite end of the spectrum, if $P_{X^T|\wh{X}^S}(\cdot|\wh{x}^S) = P_{X^T|X^S}(\cdot|x^S)$, then the causal effect of $X^T$ coincides with ordinary conditioning. This observation suggests that, for each realization $x^S$ of $X^S$, we may measure the average ``strength'' of the causal effect of the intervention $X^S \leftarrow x^S$ on $X^T$ by the divergence
\begin{align*}
	D(P_{X^T|X^S=x^S}\| P_{X^T|\wh{X}^S = \wh{x}^S}) &= \E\left[ \log\frac{P_{X^T|\wh{X}^S}(X^T|x^S)}{P_{X^T|\wh{X}^S}(X^T|\wh{x}^S)} \right]
\end{align*}
where the expectation is w.r.t.\ the conditional distribution $P_{X^T|X^S = x^S}$. If we now average this w.r.t.\ the marginal distribution of $X^S$ induced by \eqref{eq:SDM}, then we obtain
\begin{align}
& D(P_{X^T|X^S}\| P_{X^T|\wh{X}^S}|P_{X^S}) = \E \left[ \log \frac{P_{X^T|X^S}(X^T|X^S)}{P_{X^T|\wh{X}^S}(X^T|\wh{X}^S)}\right],\label{eq:directed_divergence}
\end{align}
where $D(P_{B|A}\|Q_{B|A}|P_A)$ denotes the {\em conditional divergence} \cite{CsiKor81}. If $T = S^c$, then we have
\begin{align*}
	& D(P_{X^{S^c}|X^S}\| P_{X^{S^c}|\wh{X}^S}|P_{X^S}) \\
		&\qquad  = \E \left[ \log \frac{P_{X^{S^c}|X^S}(X^{S^c}|X^S)}{P_{X^{S^c}|\wh{X}^S}(X^{S^c}|\wh{X}^S)}\right] \\
		&\qquad = \E \left[\log \frac{P_{X^{S^c}|X^S}(X^{S^c}|X^S)}{\vec{P}_{X^{S^c}|X^S}(X^{S^c}|X^S)}\right],
\end{align*}
where the second step uses the equivalence between the interventional distribution $P_{X^{S^c}|\wh{X}^S}$ and the directed stochastic kernel $\vec{P}_{X^{S^c}|X^S}$. We can now recognize the last expression as the {\em directed information} $I(X^{S^c} \to X^S)$ from $X^{S^c}$ to $X^S$ as defined by Tatikonda and Mitter \cite[p.~327]{TatikondaMitter}. This definition, in turn, generalizes the one proposed by Massey \cite{MasseyDirInfo} in the context of communication over noisy channels with feedback. Thus, directed information arises naturally as an information-theoretic measure of causality: if $I(X^{S^c} \to X^S)$ is small, then the interventional distributions of $X^{S^c}$ based on $X^S$ are close to observational (i.e., conditional) distributions of $X^{S^c}$ given $X^S$, which means that the causal effects of $X^S$ on $X^{S^c}$ can be reliably identified without the need for active experimentation. On the other hand, if $I(X^{S^c}; X^S)$ is equal to the ordinary mutual information $I(X^{S^c}; X^S)$, then the variables in $S$ have no causal effect on the remaining variables in $S^c$, and any statistical dependence between $X^S$ and $X^{S^c}$ must be along the (not necessarily directed) paths in the DAG that have some edges pointing toward $S$.

The definitions of Massey and Tatikonda--Mitter apply only to the causal effect of $X^S$ on the entire complementary set $X^{S^c}$. We can, however, consider an arbitrary set $T \subseteq S^c$ and use \eqref{eq:directed_divergence} as our definition of the directed information from $X^T$ to $X^S$:
\begin{align}\label{eq:directed_info}
	I(X^T \to X^S) \deq D(P_{X^T|X^S}\| P_{X^T|\wh{X}^S}|P_{X^S}).
\end{align}
Note that for $I(X^T \to X^S)$ to be well-defined, we need to specify an appropriate Markovian dynamical system, where the interventional distribution $P_{X^T|\wh{X}^S}$ is computed according to \eqref{eq:SDM_Markov_interv_marginal}.

An expression for the directed information $I(X^{S^c} \to X^S)$ can be obtained from the underlying graphical model. Indeed, note that we can write
\begin{align*}
	I(X^{S^c} \to X^S) = \E\left[\log \frac{P_{X^S,X^{S^c}}(X^S,X^{S^c})}{P_{X^{S^c}|\wh{X}^S}(X^{S^c}|\wh{X}^S)P_{X^S}(X^S)}\right].
\end{align*}
Now, the probability distribution in the numerator is equal to $P_{X^n}$ and can be assembled from the original Markov factorization, while the one in the denominator is the product of the interventional distribution $P_{X^{S^c}|\wh{X}^S}$ (which can be read off from the transformed DAG obtained using the procedure illustrated in Section~\ref{ssec:graphical_models}) and the marginal distribution $P_{X^S}$ according to the original model. The directed edges that are common to the original DAG and the transformed DAG correspond to the factors in the numerator and the denominator that can be cancelled. The remaining expression can then be represented as a sum of conditional mutual informations by exploiting appropriate conditional independence relations encoded in the original DAG.\footnote{We would like to thank Yury Polyanskiy for clarifications regarding this procedure.} 

\subsection{Combining interventions and passive observations: conditional directed information}

We have already pointed out the different status of active interventions of the form $X^S \leftarrow x^S$ and conditioning on passive observations $X^S = x^S$. Many problems pertaining to causality involve a combination of the two: given three disjoint sets $S,S',T \subset [n]$, we may want to consider a mixed quantity $P_{X^T|X^S \leftarrow x^S, X^{S'}=x^{S'}}$. In order for such an object to be well-defined, the conditioning on $X^{S'}$ must be done w.r.t.\ the interventional distribution of $P_{X^{S' \cup T}|X^S \leftarrow x^S}$:
\begin{align*}
	P_{X^T|X^S \leftarrow x^S, X^{S'} = x^{S'}}(x^T) \deq \frac{P_{X^{S' \cup T}|X^S \leftarrow x^S}(x^{S' \cup T})}{P_{X^{S'}|X^S \leftarrow x^S}(x^{S'})}.
\end{align*}
In fact, this is the only sensible definition, because performing the conditioning first may destroy the Markov structures that are needed to construct the interventional distribution.

With the above definition, we may define the {\em conditional directed information}
\begin{subequations}
\begin{align}
&	I(X^T \to X^S | X^{S'}) \nonumber\\
& \qquad \deq D(P_{X^T|X^S,X^{S'}} \| P_{X^T|\wh{X}^S,X^{S'}}|P_{X^S,X^{S'}}) \\
& \qquad \equiv \E \left[ \log \frac{P_{X^T|X^S,X^{S'}}(X^T|X^S,X^{S'})}{P_{X^T|\wh{X}^S,X^{S'}}(X^T|\wh{X}^S,X^{S'})}\right].
\end{align}
\label{eq:conditional_directed_info}
\end{subequations}

\subsection{Some properties of directed information}

Let us illustrate the role of the directed information \eqref{eq:directed_info} and the conditional directed information \eqref{eq:conditional_directed_info} in quantifying the causal flow of information in Markovian dynamical systems. We start with the following:

\begin{lemma} For any $S \subset [n]$ and any $T \subseteq N_S$,
	\begin{align*}
		I(X^T \to X^S) = I(X^T; X^S).
	\end{align*}
Moreover, for any $T \subseteq (S \cup \Pi_S)^c$,
\begin{align*}
	I(X^S \to X^{\Pi_S \cup T}) = I(X^S \to X^{\Pi_S}) = 0.
\end{align*}
\end{lemma}
\begin{proof} This is just a restatement of Lemmas~\ref{lm:nondescendants_2} and \ref{lm:parents} in the language of directed information.\end{proof}
We can also show that there are two contributions to the directed flow of information from $X^T$ to $X^S$: (1) the ordinary mutual information between the variables in $S$ and any nondescendants of $S$ that happen to lie in $T$, and (2) the conditional directed information from the descendants of $S$ in $T$ to $S$, given the nondescendants of $S$ in $T$:

\begin{proposition}[chain rule]\label{eq:drected_chain_rule} For any two disjoint sets $S, T \subset [n]$, we have
	\begin{align}
	&	I(X^T \to X^S) \nonumber\\
	& \quad = I(X^{T \cap N_S}; X^S) + I(X^{T \cap \Delta_S} \to X^S | X^{T \cap N_S}).\label{eq:directed_chain_rule}
	\end{align}
\end{proposition}
\begin{proof} For brevity, let us denote $T_1 = T \cap N_S$ and $T_2 = T \cap \Delta^+_S$ (which is equal to $T \cap \Delta_S$ since $S \cap T = \varnothing$). Then
\begin{align*}
&	P_{X^T|X^S \leftarrow x^S}(x^T)\\
 &= P_{X^{T_1}|X^S \leftarrow x^S}(x^{T_1})P_{X^{T_2}|X^S \leftarrow x^S, X^{T_1} = x^{T_1}}(x^{T_2}) \\
&= P_{X^{T_1}}(x^{T_1}) P_{X^{T_2}|X^S \leftarrow x^S, X^{T_1} = x^{T_1}}(x^{T_2}),
\end{align*}
where the second step uses Lemma~\ref{lm:nondescendants_2}. Similarly,
\begin{align*}
	& P_{X^T|X^S = x^S}(x^T) \\
	&= P_{X^{T_1}|X^S = x^S}(x^{T_1}) P_{X^{T_2}|X^S = x^S, X^{T_1}=x^{T_1}}(x^{T_2}).
\end{align*}
Therefore,
\begin{align*}
	I(X^T \to X^S) &= \E\left[\log \frac{P_{X^{T_1}|X^S}(X^{T_1}|X^S)}{P_{X^{T_1}}(X^{T_1})}\right] \\
	&\quad + \E \left[\log \frac{P_{X^{T_2}|X^S,X^{T_1}}(X^{T_2}|X^S,X^{T_1})}{P_{X^{T_2}|\wh{X}^S,X^{T_1}}(X^{T_2}|\wh{X}^S,X^{T_1})} \right] \\
	&= I(X^{T_1}; X^S) + I(X^{T_2} \to X^S | X^{T_1}),
\end{align*}
which gives us \eqref{eq:directed_chain_rule}.
\end{proof}
\begin{corollary} For any set $S \subset [n]$,
	\begin{align*}
		I(X^{S^c} \to X^S) = I(X^{N_S}; X_S) + I(X^{\Delta_S} \to X^S | X^{N_S}).
	\end{align*}
\end{corollary}
\begin{proof} Immediate from the proposition with $T = S^c$.\end{proof}
	
	\subsection{Examples: three canonical causal structures}

	Many fundamental questions pertaining to causality (including the possibility of discovering causal influences from observational data) can be reduced to the study of three canonical causal structures involving three random variables $X,Y,Z$: the {\em chain} $X \to Y \to Z$; the {\em fork} $X \leftarrow Y \rightarrow Z$; and the {\em collider} $X \to Y \leftarrow Z$ \cite{SpirtesBook,PearlCausality}. We have the following examples of directed information relations for these structures:\\

	\noindent{\bf Chain.} Since $X$ is a nondescendant of $Y$, we have $I(X \to Y) = I(X;Y)$; since $X$ is the direct cause of $Y$ we have $I(Y \to X) = 0$. Similarly, we have $I(Y \to Z) = I(Y; Z)$ and $I(Z \to Y) = 0$. Moreover, since $X$ is a nondescendant of $Z$, we have $I(X \to Z) = I(X;Z)$. On the other hand, $I(Z \to X) = 0$.\\

	\noindent{\bf Fork.} $Y$ is the direct cause of $X$, so $I(X \to Y) = 0$, and it is a nondescendant of $X$, so $I(Y \to X) = I(X;Y)$. The same goes for $Y$ and $Z$: $I(Z \to Y) = 0$ and $I(Y \to Z) = I(Y;Z)$. Finally, we have $I(X \to Z) = I(Z \to X) = I(X;Z)$, since there is no directed path from $X$ to $Z$ or from $Z$ to $X$.\\	
	
	\noindent{\bf Collider.} The direction of the links between $X$ and $Y$ and between $Z$ and $Y$ is the reverse of that in the fork, so we have $I(X \to Y) = I(X;Y)$, $I(Y \to X) = 0$, $I(Y \to Z) = 0$, and $I(Z \to Y) = I(Y;Z)$. Finally, since $X$ is a nondescendant of $Z$, we have $I(X \to Z) = I(X; Z) = 0$; similarly, $I(Z \to X) = I(X;Z) = 0$, where we have also used the fact that $X$ and $Z$ are independent.

\section{Application to identification of causal effects}
\label{sec:back_door}

One active area of interest in the studies of causality concerns identification of causal effects based on passive observations only. In the context of Markovian dynamical system models, this problem arises whenever only a subset of the variables $X^n$ is available for observation, the goal is to determine the causal effect of one group of variables in this subset upon another, and it is not possible or feasible to actively intervene into the system. Then the relevant question becomes: given a set $V \subset [n]$ that indexes the variables available for observation, is it possible to express the causal effect $P_{X^T|\wh{X}^S}$ for some disjoint sets $S,T \subset V$ in terms of ordinary (noninterventional) probabilities? 

More precisely, let us assume that we know the structure of the underlying DAG (i.e., the sets $\Pi_i, i \in [n]$), but not the functions $f_i$ or the distributions $P_{U_i}$ of the exogenous disturbances. What other variables besides those in $S$ and $T$ do we need to observe in order to estimate the causal effect $P_{X^T|\wh{X}^S}$? The idea is that the ordinary conditional probabilities relating the variables in $V$ can be estimated from passive observations, and so $P_{X^T|\wh{X}^S}$ can be estimated using a plug-in rule in terms of these conditional probabilities.

One obvious answer is that it is sufficient to observe $S$, $T$, and all direct causes of the variables in $S$, i.e., those in $\Pi_S$. To see this, let us write down the interventional distribution $P_{X^T|X^S \leftarrow x^S}$ and condition on $X^{\Pi_S}$:
\begin{align*}
	& P_{X^T|X^S \leftarrow x^S}(x^T)\nonumber\\
	 &\quad = \sum_{x^{\Pi_S}} P_{X^T|X^S \leftarrow x^S ,X^{\Pi_S} = x^{\Pi_S}}(x^T)P_{X^{\Pi_S}|X^S \leftarrow x^S}(x^{\Pi_S}) \\
	&\quad = \sum_{x^{\Pi_S}} P_{X^T|X^S \leftarrow x^S ,X^{\Pi_S} = x^{\Pi_S}}(x^T)P_{X^{\Pi_S}}(x^{\Pi_S}),
\end{align*}
where the second step uses the fact that $\Pi_S \subseteq N_S$ and Lemma~\ref{lm:nondescendants_2}. Now, it can be shown that $P_{X^T|X^S \leftarrow x^S, X^{\Pi_S} = x^{\Pi_S}} = P_{X^T|X^S = x^S,X^{\Pi_S}=x^{\Pi_S}}$ \cite[Thm.~3.2.2]{PearlCausality}, which is equivalent to $I(X^T \to X^S | X^{\Pi_S}) = 0$. This gives
\begin{align}
	&P_{X^T|X^S \leftarrow x^S}(x^T) \nonumber\\
	& \qquad = \sum_{x^{\Pi_S}}P_{X^T|X^S = x^S, X^{\Pi_S} = x^{\Pi_S}}(x^T)P_{X^{\Pi_S}}(x^{\Pi_S}). \label{eq:direct_causes}
\end{align}
Thus, if we observe the variables in $T$, $S$, and $\Pi_S$, then we can use \eqref{eq:direct_causes} to develop an estimate of the causal effect $P_{X^T|\wh{X}^S}$ in terms of the conditional distribution $P_{X^T|X^S,X^{\Pi_S}}$ and the marginal distribution $P_{X^{\Pi_S}}$. Both of these quantities can, in turn, be estimated from passive observations. The intuitive meaning of \eqref{eq:direct_causes} is that we can estimate the causal effect of $X^S$ on $X^T$ without any need for active experimentation if we can control for the direct causes of $X^S$, i.e., $X^{\Pi_S}$. Whenever this is not possible, we would still like to know what other variables it suffices to observe in order for the causal effect $P_{X^T|\wh{X}^S}$ to be identifiable. One sufficient condition due to Pearl, who termed it the ``back-door criterion'' \cite[Sec.~3.3.1]{PearlCausality}, says that certain subsets of the nondescendants of $S$ can be used instead:

\begin{theorem}[the back-door criterion: directed information form] Let $S,T \subset [n]$ be such that $T$ is disjoint from $S \cup \Pi_S$. Then for any set $Z \subseteq N_S$ the relation
	\begin{align}
		&P_{X^T|X^S \leftarrow x^S}(x^T) \nonumber\\
		&\qquad = \sum_{x^Z} P_{X^T|X^{S \cup Z} = x^{S \cup Z}}(x^T)P_{X^Z}(x^Z)
	\end{align}
	holds if and only if $I(X^T \to X^S | X^Z) = 0$.
\end{theorem}

\begin{proof} Let us condition on $X^Z$:
	\begin{align*}
	&	P_{X^T|X^S \leftarrow x^S}(x^T) \\
	&= \sum_{x^Z} P_{X^T|X^S \leftarrow x^S, X^Z = x^Z}(x^T)P_{X^Z|X^S \leftarrow x^S}(x^Z) \\
	&= \sum_{x^Z} P_{X^T|X^S \leftarrow x^S, X^Z = x^Z}(x^T)P_{X^Z}(x^Z),
	\end{align*}
	where the second step uses the fact that $Z \subseteq N_S$ and Lemma~\ref{lm:nondescendants_2}. The proof is finished using the fact that $P_{X^T|X^S \leftarrow x^S, X^Z = x^Z} = P_{X^T|X^{S \cup Z}=x^{S \cup Z}}$ for all $x^S, x^Z$ if and only if $I(X^T \to X^S | X^Z) = 0$.
\end{proof}
\noindent The original back-door criterion \cite[Section~3.3.1]{PearlCausality} is stated in graphical terms using the notion of {\em d-separation} (a graph-based criterion for identifying conditional independence relations), so it can be checked without knowing $\{f_i\}^n_{i=1}$ or $\{P_{U_i}\}^n_{i=1}$. Conceptually, its equivalent information-theoretic form given by the above theorem is similar to statistical sufficiency: if $Z \subseteq N_S$, then $X^Z$ may only depend functionally on $X^T$ (but not on $X^S$ or on any of the descendants of $X^S$), and if $I(X^S ; X^T | X^Z) = 0$, then $X^Z$ is sufficient for $X^S$ in the ordinary Bayesian sense.

\section*{Acknowledgment}

The author would like to thank Todd Coleman, Prakash Ishwar, Tara Javidi, Donatello Materassi, and Yury Polyanskiy for stimulating discussions.

\bibliographystyle{IEEEtran}
\bibliography{Pearl.bbl}

\begin{thebibliography}{10}
\providecommand{\url}[1]{#1}
\csname url@samestyle\endcsname
\providecommand{\newblock}{\relax}
\providecommand{\bibinfo}[2]{#2}
\providecommand{\BIBentrySTDinterwordspacing}{\spaceskip=0pt\relax}
\providecommand{\BIBentryALTinterwordstretchfactor}{4}
\providecommand{\BIBentryALTinterwordspacing}{\spaceskip=\fontdimen2\font plus
\BIBentryALTinterwordstretchfactor\fontdimen3\font minus
  \fontdimen4\font\relax}
\providecommand{\BIBforeignlanguage}[2]{{%
\expandafter\ifx\csname l@#1\endcsname\relax
\typeout{** WARNING: IEEEtran.bst: No hyphenation pattern has been}%
\typeout{** loaded for the language `#1'. Using the pattern for}%
\typeout{** the default language instead.}%
\else
\language=\csname l@#1\endcsname
\fi
#2}}
\providecommand{\BIBdecl}{\relax}
\BIBdecl

\bibitem{RaoMotifDiscovery}
A.~Rao, A.~O. {Hero III}, D.~J. States, and J.~D. Engel, ``Motif discovery in
  tissue-specific regulatory sequences using directed information,''
  \emph{{EURASIP} J. Bioinf. Sys. Biol.}, 2007, art. no. 13853.

\bibitem{MathaiGeneNetworks}
P.~Mathai, N.~C. Martins, and B.~Shapiro, ``On the detection of gene network
  interconnections using directed mutual information,'' in \emph{Proc. Inform
  Th. Appl. Workshop}, La Jolla, CA, January/February 2007, pp. 274--283.

\bibitem{RaoBioNetworks}
A.~Rao, A.~O. {Hero III}, D.~J. States, and J.~D. Engel, ``Using directed
  information to build biologically relevant influence networks,'' \emph{J.
  Bioinf. Comput. Biol.}, vol.~6, no.~3, pp. 493--519, 2008.

\bibitem{AmblardNeuro}
P.-O. Amblard and O.~J.~J. Michel, ``On directed information and {G}ranger
  causality graphs,'' \emph{J. Comput. Neurosci.}, vol.~30, no.~1, pp. 7--16,
  2011.

\bibitem{QuinnNeuro}
C.~J. Quinn, T.~P. Coleman, N.~Kiyavash, and N.~G. Hatsopoulos, ``Estimating
  the directed information to infer causal relationships in ensemble neural
  spike train recordings,'' \emph{J. Comput. Neurosci.}, vol.~30, no.~1, pp.
  17--44, 2011.

\bibitem{QuinnCausalTrees}
\BIBentryALTinterwordspacing
C.~J. Quinn, T.~P. Coleman, and N.~Kiyavash, ``Causal dependence tree
  approximations of joint distributions for multiple random processes,''
  \emph{IEEE Trans. Inform. Theory}, 2011, submitted. [Online]. Available:
  \url{http://arxiv.org/abs/1101.5108}
\BIBentrySTDinterwordspacing

\bibitem{MasseyDirInfo}
J.~Massey, ``Causality, feedback, and directed information,'' in \emph{Proc.
  Int. Symp. Inf. Theory Appl.}, 1990, pp. 303--305.

\bibitem{KramerThesis}
G.~Kramer, ``Directed information for channels with feedback,'' Ph.D.
  dissertation, Swiss Federal Institute of Technology, Zurich, Switzerland,
  1998.

\bibitem{TatikondaThesis}
S.~Tatikonda, ``Control under communication constraints,'' Ph.D. dissertation,
  MIT, Cambridge, MA, August 2000.

\bibitem{TatikondaMitter}
S.~Tatikonda and S.Mitter, ``The capacity of channels with feedback,''
  \emph{IEEE Trans. Inform. Theory}, vol.~53, no.~1, pp. 323--349, 2009.

\bibitem{VenkataramananFeedforward}
R.~Venkataramanan and S.~S. Pradhan, ``Source coding with feedforward:
  rate-distortion theorems and error exponents for a general source,''
  \emph{IEEE Trans. Inform. Theory}, vol.~53, no.~6, pp. 2154--2179, June 2007.

\bibitem{PermuterTrapdoor}
H.~Permuter, P.~Cuff, B.~{Van Roy}, and T.~Weissman, ``Capacity of the trapdoor
  channel with feedback,'' \emph{IEEE Trans. Inform. Theory}, vol.~54, no.~7,
  pp. 3150--3165, July 2008.

\bibitem{GorantlaColemanCausal}
\BIBentryALTinterwordspacing
S.~Gorantla and T.~Coleman, ``Information-theoretic viewpoints on optimal
  causal coding-decoding problems,'' \emph{IEEE Trans. Inform. Theory}, 2011,
  submitted. [Online]. Available: \url{http://arxiv.org/abs/1102.0250}
\BIBentrySTDinterwordspacing

\bibitem{PermuterDirInfo}
H.~H. Permuter, Y.-H. Kim, and T.~Weissman, ``Interpretations of directed
  information in portfolio theory, data compression, and hypothesis testing,''
  \emph{IEEE Trans. Inform. Theory}, vol.~57, no.~6, pp. 3248--3259, June 2011.

\bibitem{PearlPRIS}
J.~Pearl, \emph{Probabilistic Reasoning in Intelligent Systems}.\hskip 1em plus
  0.5em minus 0.4em\relax San Francisco, CA: Morgan Kaufmann, 1988.

\bibitem{SpirtesBook}
P.~Spirtes, C.~Glymour, and R.~Scheines, \emph{Causation, Prediction, and
  Search}, 2nd~ed.\hskip 1em plus 0.5em minus 0.4em\relax MIT Press, 2000.

\bibitem{PearlCausalitySurvey}
J.~Pearl, ``Causal inference in statistics: an overview,'' \emph{Statist.
  Surv.}, vol.~3, pp. 96--146, 2009.

\bibitem{PearlCausality}
------, \emph{Causality: Models, Reasoning, and Inference}, 2nd~ed.\hskip 1em
  plus 0.5em minus 0.4em\relax Cambridge Univ. Press, 2009.

\bibitem{WitsenhausenInfoStruct}
H.~S. Witsenhausen, ``On information structures, feedback and causality,''
  \emph{{SIAM} J. Control}, vol.~9, no.~2, pp. 149--160, 1971.

\bibitem{WitsenhausenSep}
------, ``Separation of estimation and control for discrete time systems,''
  \emph{Proc. IEEE}, vol.~59, no.~11, pp. 1557--1566, November 1971.

\bibitem{WitsenhausenStandardSequential}
------, ``A standard form for sequential stochastic control,'' \emph{Math. Sys.
  Theory}, vol.~7, no.~1, pp. 5--11, 1973.

\bibitem{WrightGenetics}
S.~Wright, ``Correlation and causation,'' \emph{J. Agricultural Res.}, vol.~20,
  pp. 557--585, 1921.

\bibitem{Dob59}
R.~L. Dobrushin, ``A general formulation of the basic {S}hannon theorem in
  information theory,'' \emph{Uspekhi Math. Nauk}, vol.~14, no.~6, pp. 3--103,
  1959.

\bibitem{CsiKor81}
I.~Csisz\'ar and J.~K\"orner, \emph{Information Theory: Coding Theorems for
  Discrete Memoryless Sources}.\hskip 1em plus 0.5em minus 0.4em\relax
  Budapest: Akad\'emiai Kiad\'o, 1981.

\end{thebibliography}
\end{document}